\documentclass[12pt]{article}

\usepackage{epsf}
\usepackage{amsbsy}
\usepackage{epsfig}
\usepackage{amssymb}
\usepackage{mathrsfs}
\usepackage{latexsym}
\usepackage{enumerate}
\usepackage{amsmath}
\usepackage{amsfonts}

\newcommand{\remove}[1]{}
\newtheorem{lemma}{Lemma}[section]

\newtheorem{theorem}{Theorem}[section]
\newtheorem{corollary}{Corollary}[section]
\newtheorem{definition}{Definition}[section]

\newcommand{\email}[1]{E-mail: #1}
\newcounter{pclaim} 

\newcommand{\qedsymb}{\hfill{\rule{2mm}{2mm}}}

\newenvironment{proof}{\begin{trivlist}
\item[\hspace{\labelsep}{\bf\noindent Proof: }]
}{\qedsymb\end{trivlist}}

\begin{document}

\title{\LARGE\bf Improved lower bound for deterministic broadcasting in radio networks}

\author{
\sc{Carlos Fisch Brito} \thanks{Departamento de Computa\c{c}\~ao, Universidade Federal do Cear\'a, Fortaleza, Brazil.  \email{carlos@lia.ufc.br}. Research in this work was conducted while the author was a graduate student at UCLA.}
\and \sc{Shailesh Vaya} \thanks{Department of Computer Science and Engineering, Indian Institute of Technology Madras, Chennai, India - 600036. \email{vaya@cse.iitm.ernet.in}. Research in this work was conducted while the author was a graduate student at UCLA.}
}

\maketitle


\abstract
{
  We consider the problem of deterministic broadcasting in radio networks when the nodes have limited knowledge about the topology of the network. We show that for every deterministic broadcasting protocol there exists a network, of radius $2$, for which the protocol takes at least $\Omega(n^{\frac{1}{2}})$ rounds for completing the broadcast. Our argument can be extended to prove a lower bound of $\Omega((nD)^{\frac{1}{2}})$ rounds for broadcasting in radio networks of radius $D$. This resolves one of the open problems posed in \cite{KP02}, where in the authors proved a lower bound of $\Omega(n^{\frac{1}{4}})$ rounds for broadcasting in constant diameter networks.

  We prove the new lower bound for a special family of radius $2$ networks. Each network of this family consists of $O(\sqrt{n})$ components which are connected to each other via only the source node. At the heart of the proof is a novel simulation argument, which essentially says that any arbitrarily complicated strategy of the source node can be simulated by the nodes of the networks, if the source node just transmits partial topological knowledge about some component instead of arbitrary complicated messages. To the best of our knowledge this type of simulation argument is novel and may be useful in further improving the lower bound or may find use in other applications.
}


\newpage
\pagenumbering{arabic}

\section{Introduction}
\label{sec:introduction}

  One of the most fundamental and well studied problems in distributed computing is the problem of broadcasting. In broadcasting there is a {\it source node} which possesses a message that should be sent to all the remaining nodes of the network. If the source node is directly connected to all the nodes of the network then the message can be sent to the rest of the nodes in a single transmission. What makes the problem of broadcasting non-trivial is that the source node may not be directly connected to the rest of the network and the topology of the network may be arbitrary. The minimum requirement on the topology of a network so that broadcasting can be completed is that the network is connected. There are several metrics to measure the complexity of a broadcasting protocol, like round complexity (minimum number of rounds needed to complete broadcasting), message complexity (minimum number of messages needed to be sent to complete broadcasting) etc., which can be defined in terms of the number of nodes in the network $n$, the radius of the network $D$, etc. The most important and prevalent complexity measure under which the broadcasting problem has been studied is the {\it round complexity} of the broadcasting protocol. This work focuses on the round complexity of deterministic broadcasting protocols in special kinds of networks called {\it radio networks}.

{\bf Ad-hoc radio networks.}
  Broadcasting protocols find many applications in {\it Ad-hoc wireless} networks. Ad-hoc wireless networks are used in scenarios like battlefields, emergency disaster reliefs and other situations where there is no infrastructure for communication networks. Sensor networks are also an example of wireless radio networks. Unlike the traditional wireless networks these networks do not have a base station to which the nodes can communicate. Nodes which are within the range of their radio signals communicate via radio transmissions, while nodes that are far off rely on other nodes of the network to exchange messages.

  Communication in these networks is structured using synchronous time-slots. In every round each node either acts as a transmitter or as a receiver. A radio network can be modeled as an {\em undirected} connected graph as follows: Each node in the graph represents a processor, and two nodes are connected by an edge if the corresponding processors lie within the transmission range of each other. A message transmitted by a node can potentially reach all its neighbors.  However, if more than one neighboring node send a message in the same round, then a collision occurs. This is the model that was proposed in the seminal paper on radio networks \cite{BGI}, and has been generally considered in the literature \cite{KP02}. In \cite{BGIE}, the authors adopt a more pessimistic model for radio transmission, according to which when two or more neighboring nodes transmit a message in the same round, the receiving node receives the message from one of them and the messages from others are lost. The model considered in this work is the one traditionally considered in the literature on radio broadcasting, i.e. as stated in \cite{BGI}, \cite{KP02}, \cite{KP04}.

\subsection{Related Work}
\label{sec:related-work}
  The study of broadcasting in radio networks was initiated by Bar-Yehuda, Goldreich and Itai in \cite{BGI}. In \cite{KP02}, Kowalski and Pelc establish a lower bound of $\Omega (n^{1/4})$ rounds for the deterministic broadcasting problem on radio networks of diameter $4$ (and $(nD^3)^{1/4}$ rounds for networks of diameter $D$) in the model considered traditionally, i.e., which is the model stated in \cite{BGI}.

  When the topology of the radio network is known to all the nodes of the network, a deterministic broadcasting protocol of $O(D \lg_2^2 n)$ was given in \cite{CW}, for networks of radius $D$. This centralized broadcasting protocol has recently been improved to $O(D \lg_2 n + \lg_2^2 n)$ rounds by \cite{KP04}. A protocol running in $O(D + \lg^5_2 n)$ rounds was given in \cite{GM} and this has been most recently improved to $O(D + \lg_2^4 n)$ in \cite{EK}. A lower bound of $\Omega(\lg_2 n)$ rounds has been proved in \cite{ANLP}, for the centralized setting.

  The setting where the nodes are given only their own labels, but not the labels of their neighboring nodes has also been studied quite extensively. Research in this problem has led to the introduction of very interesting combinatorial concepts like {\it selective families}.  The use of selective families in the design of deterministic protocols for unknown networks was introduced by Chlebus et. al. in \cite{CGGPR}. Several recent works exploit this combinatorial tool, specifically the use of probabilistic method, for obtaining good lower and upper bounds for the broadcasting problem \cite{CMS}, \cite{CGGPR}, \cite{CGR}.

  The problem of broadcasting on directed radio networks has received much attention. The protocol given by \cite{CGGPR} requires $O(n^2)$ rounds for completion. This upper bound was reduced to $O(n^\frac{3}{2})$ rounds by a breakthrough result of \cite{CGOR}. In a further breakthrough, \cite{CGR}, the authors deployed the probabilistic method to bring the upper bound from $O(n \lg^2_2 n)$ to within logarithmic factors of the best known lower bound $\Omega(n \lg_2 D)$. In \cite{KP02}, a further improvement has been made for small diameter networks. They establish an upper bound of $O(n \lg_2 n \lg_2 D)$. In \cite{CR}, the authors define and deploy new combinatorial structures to reduce this gap between lower and upper bounds to $O(\lg_2 D)$ factor i.e., give a broadcasting protocol running in $O(n \lg^2_2 D)$ rounds.

  Randomized protocols for the broadcasting problem have been studied in \cite{BGI}, \cite{KP3}, \cite{KM}. In these protocols, the nodes do not know the labels of their neighboring nodes and may in fact have non-unique labels. In \cite{BGI}, a broadcasting protocol running in $O(D \lg_2 n + \lg_2^2 n)$ rounds is given. A lower bound of $\Omega(D \lg_2(N/D))$ rounds was proved for this problem in \cite{KM}. The lower bound of $\Omega(\lg_2^2 n)$ rounds for broadcasting protocols also holds for this problem. A broadcasting protocol running in $O(D \lg_2 (n/D) + \lg_2^n)$, matching the lower bound, was proposed in \cite{KP3}.

  In the wake up problem, each node in the network either wakes-up spontaneously or is activated by receiving a wake-up signal from another node. Each active node transmits the wake-up signal according to a given protocol. The running time of a wake-up protocol is the number of steps counted from the first spontaneous wake-up, until all nodes become activated. This problem has been studied under different assumptions in \cite{CGK}, \cite{CK}, \cite{GPP}, \cite{JS}. Asynchronous radio broadcasting has been studied in \cite{CRM}.  

  In the next section, we shall present an outline of the proof of the lower bound. The rest of the paper shall be devoted to formalizing this outline.

\section{Some early developments and outline of the Proof}

  The study of broadcasting in radio networks was initiated by Bar-Yehuda, Goldreich and Itai in \cite{BGI}, where they constructed a class of networks, referred to as simple BGI-networks (see Figure \ref{fig:bgi-c2}(a)), for which they presented a lower bound for deterministic broadcasting. In \cite{KP02}, Kowalski and Pelc observed that for simple BGI-networks the role of the source node cannot be ignored in achieving faster broadcast. They go on to describe a deterministic protocol in which the source assists other nodes in the network to complete broadcast on any BGI-network in $O(\lg_2 n)$ rounds. 

\begin{figure}[h]
   \centerline{
      \epsfig{file=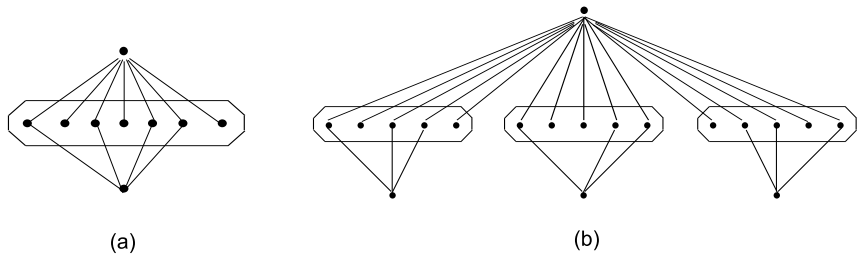, scale=0.5}
   }
   \caption{(a) BGI-network and (b) C2 network}
   \label{fig:bgi-c2}
\end{figure}

  For the result in \cite{BV03}, \cite{BGV04}, we introduced a more elaborate class of networks, (see Figure \ref{fig:bgi-c2}(b)) for which we asked the questions: "If the broadcast protocol runs for at most $r$ rounds, how much assistance can the source node provide? Can this assistance be quantified in terms of the number of rounds, the topology of the network and argued to be insufficient to significantly expedite broadcasting?". 

  The first attempt to formalize this idea was an information theoretic argument presented in \cite{BV03}. Basically it is argued in \cite{BV03} that even if the role of the source node is substituted by an advice string encoding some partial knowledge about the topology of the underlying network in consideration, the assistance is found to be insufficient, information theoretically speaking, to complete broadcast in $\sqrt{n}$ rounds on at least some of the networks. This argument employs a straightforward generalization of a lower bound on the size of selective families, \cite{CMS}. In this generalization we found a technical gap. The generalization is perhaps true for families of subsets with special properties, but we believe that formulating an argument for the lower bound along the original lines would make the proof unnecessarily complicated. Along the same lines an oracle argument, \cite{MOLLE}, was also sketched in \cite{BGV04}. This argument has the conceptual problem to ameliorate which the original information theoretic argument was proposed in \cite{BV03}. Namely, if an oracle provides the identities of the two nodes with least IDs which transmitted in a given round when a collision is detected, then depending on the specification of the protocol the information provided by the oracle can also leak partial knowledge about the topology of the network. For example, it may be the case that the protocol specifies that some nodes in layer $L_1$, with IDs smaller than the ones revealed by the oracle also transmit in this round if connected to the appropriate node in layer $L_2$ (see Figure 2). Then, an interpretation of the oracle message is that these nodes are not connected to the corresponding nodes in $L_2$ in the given network. This knowledge about the topology of the network may be exploited in achieving faster broadcast in future and cannot just be ignored.

  In this work we present a more direct proof of the lower bound, along lines similar to the original information theoretic argument. The proof consists of three main components:\\

  In the first component it is shown through a series of simulation based reductions that, if an arbitrary deterministic broadcast protocol $\pi_0$ completes broadcast in $r$ rounds on every $C_2$ network, then there must exist another deterministic broadcast protocol in which the source is given an advice string at the beginning of the protocol $\pi'$ which the source transmits along with the broadcast message at the beginning of the protocol and remains silent for the rest of the protocol. This advice string consists of $r$ blocks of binary strings each of which encodes one of the following two types of messages: $\phi$ (which denotes a collision or empty message), or a tuplet $<i,C_i>$ where $C_i$ denotes the configuration of the $i^{th}$ BGI-component in the network. The protocol $\pi'$ must complete broadcast in at most $3*r$ rounds on every network of the $C_2$ family. These reductions are presented in Section \ref{sec:reductions}.

  The essence of the second component is to utilize meaningful connections between the messages encoded in the advice string and the topology of the corresponding underlying network. For this we consider the set of all the unique advice strings provided to the source node for different networks from $C_2$ family. We select one such advice string which must be provided to the source node for a large subset of $C_2$ networks. This subset of networks has the property that all the different possible configurations of at least one BGI-component are present in some network of the subset. These arguments are presented in Section \ref{sec:prune}.

  The above two components are combined with a known lower bound on the size of selective families, or alternatively a hitting game argument from \cite{BGI}, to complete the proof of the lower bound, Theorem \ref{thm:main-theorem}.

\subsection{Organization of the rest of the paper}
  The rest of the paper is devoted to formalizing the above argument. Section \ref{sec:model-definitions} describes the model in detail and introduces notations and definitions used in the text. Section \ref{sec:lower-bound} states some auxiliary results and proves the main theorem.  The statement of the reductions and their proofs are presented in Section \ref{sec:reductions}. Section \ref{sec:prune} describes the procedure PRUNE() that chooses an appropriate advice string and a corresponding subset of $C_2$ networks as mentioned above.

\section{Description of the model and general definitions}
\label{sec:model-definitions}

{\bf Definition} \cite{BGI},\cite{KP02}:
   A broadcast protocol $\pi$ for a radio network is a synchronous multiprocessor protocol which proceeds in rounds as follows:\\

\begin{enumerate}
\item
   Nodes have distinct labels from the set $\{0,1,\dots,m\}$, where $m$ is a polynomial on the number of nodes in the network. A distinguished node with label 0 is called the {\it source} node.

\item
   All nodes execute identical copies of the same protocol $\pi$.

\item
   In each round, every node either acts as a transmitter or as a receiver (or is inactive).

\item
   A node receives a message in a specific round if and only if it acts as a receiver 
  and exactly one of its neighbors transmits in that round. Otherwise, it receives $\phi$.
  We assume that the messages are authenticated, that is, when a node receives a message 
  it knows the label of the transmitting node.

\item
   The action of a node in a specific round is determined by

   \begin{enumerate}
   \item
      Initial input, which contains its own label and the labels of its neighbors.
   \item
      Messages received by the node in previous rounds.
   \end{enumerate}

\item
   In round $0$, only the source node transmits the broadcast message.

\item
   Only nodes that have received a message are allowed to transmit. That is, the only 
  {\it spontaneous} transmission is the one by the source in round $0$.

\item
   Broadcast is completed in $r$ rounds if all the nodes receive the source message in one 
  of the rounds $0,1,\dots,r-1$.

\end{enumerate}

%

\subsection{${\cal C}_2$ networks}
\label{subsec:C2networks}

\begin{figure}[h]
   \centerline{
      \epsfig{file=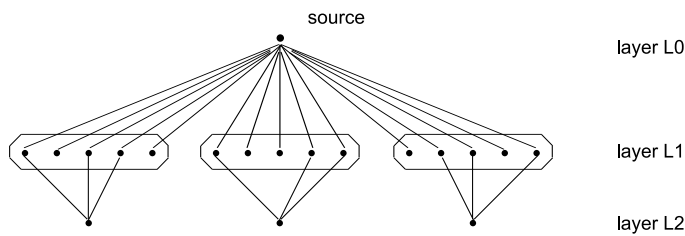, scale=0.5}
   }
   \caption{C2 network}
   \label{fig:c2-network}
\end{figure}

   We prove the lower bound for a family of networks, called $C_2$ networks, with radius $2$ and a simple communication structure. Intuitively, a $C_2$ network is formed by connecting a number of BGI-networks via a common source node (see Figure \ref{fig:c2-network}). More precisely, the nodes of a $C_2$ network are divided into three layers: $L_0$, $L_1$ and $L_2$, where:\\

\begin{itemize}
\item
   layer $L_0$ consists only of the source node;

\item
   layer $L_1$ consists of $\sqrt{n}$ groups of $\sqrt{n}$ nodes each, all of which are connected to the source node.

\item
   layer $L_2$ consists of $\sqrt{n}$ nodes, each one of which is associated with a distinct group of nodes in $L_1$ and connected to an arbitrary subset of nodes of the group.

\end{itemize}

  We call each group of nodes in $L_1$ together with the associated node in $L_2$ a BGI-component, or alternatively a component, of the network.

  The labels of the nodes in a network are arbitrary but fixed for all the $C_2$ networks. Thus, what distinguishes two networks $N$ and $M$ of $C_2$ family is only the topology of one or more BGI-components. Observing that there exists $2^{\sqrt{n}}$ distinct possible topologies for a BGI-component, we can give a complete description of a component by a tuplet $<i,\tau>$, where $i$ is an index for the component which is being described and $\tau$ is an integer in the range $[0 \dots 2^{\sqrt{n}} - 1]$.

\subsection{Advice string $\upsilon(\pi,N,t)$}
\label{subsec:advicearray}

  An advice string $\upsilon(\pi,N,t)$ is an array of $t$ elements, each of which could potentially encode partial topological information about the network $N$. More precisely, each element consists of either $\phi$ or a tuplet $<i,\tau_i>$ which describes the topology of a BGI-component of the network N. The specific contents of the advice string $\upsilon(\pi,N,t)$ depends on the first $t$ rounds of the execution of protocol $\pi$ on network $N$.


\section{Main theorem}
\label{sec:lower-bound}

  In this section we present the main theorem \ref{thm:main-theorem} which proves the lower bound. Theorem \ref{thm:main-theorem} combines the three parts of the argument Lemma \ref{lem:reduction} (simulation based reductions), Lemma \ref{lem:prune} (selection of an advice string and an appropriate subset of networks) and Theorem \ref{thm:sel-family} (a well known lower bound on the size of selective families).

\begin{lemma}
\label{lem:reduction}
  Assume that there exists a protocol $\pi$ that completes broadcast in at most $r$ rounds on every network of the $C_2$ family. Then, there exists a deterministic broadcast protocol $\pi'$ such that

  \begin{itemize}
  \item
     the nodes in layer $L_i$ transmit only in rounds $t \equiv i$ (mod 3), for $i = 0, 1, 2$.

  \item
     for each network $N$ of the $C_2$ family there is an advice string $\upsilon(\pi_3,N,r)$ which is provided to the source node at the set up phase along with the broadcast message.

  \item
     the source transmits the broadcast message and the string $\upsilon(\pi_3, N, 3r)$ in round $0$.
  \item
     the source remains silent in every round $i > 0$.

  \item
     the protocol $\pi'$ completes broadcast on network $N$ in at most $3*r$ rounds
  \end{itemize}
\end{lemma}

\begin{lemma}
\label{lem:prune}
   Let $\pi'$ be the protocol given by lemma \ref{lem:reduction}, with $r < \frac{\sqrt{n}}{2}$. Then, there exists an advice string $\upsilon$ and a subset of networks $S \subset C_2$ such that

  \begin{enumerate}
  \item
     Protocol $\pi'$ completes broadcast in at most $3*r$ rounds on every network of subset $S$, when the source node is provided with the advice string $\upsilon$ in the set up phase.

  \item
     There exists an index $i$ such that, for each possible topology of a BGI-component, there is a network $N \in S$ whose $i^{th}$ component has exactly this topology.
  \end{enumerate}
\end{lemma}

\begin{definition}[\cite{CMS}]
\label{def:sel-family}
  Let $[n] = \{1, \dots, n\}$ and let $k \leq n$. A family ${\cal F}$ of subsets of $[n]$ is a $(n,k)$-selective if for every non empty subset $Z$ of $[n]$ such that $|Z| \leq k$, there is a set $F$ in ${\cal F}$ such that $|Z \bigcap F| = 1$.
\end{definition}

\begin{theorem}[\cite{CMS}]
\label{thm:sel-family}
  Let ${\mathcal F}$ be a $(n,k)$-selective family, with $n > 2$ and $2 \leq k \leq n/64$. Then it holds that, $|{\cal F}| \geq \frac{k}{24} \log_2 \frac{n}{k}$.
\end{theorem}

\begin{theorem}
\label{thm:main-theorem}
  For every deterministic protocol $\pi$ that runs for $o(\sqrt{n})$ rounds, there exists a $C_2$ network $N$ such that $\pi$ does not complete broadcast when executed on $N$.
\end{theorem}

\begin{proof}
\label{prf:main-proof}
  Assume, to the contrary, that there exists a broadcast protocol $\pi$ that completes broadcast in less than $\sqrt{n}/1536$ rounds on every $C_2$ network. Then by applying lemma \ref{lem:reduction} we obtain a protocol $\pi'$. Let $\upsilon,S,i$ be the advice string, subset of networks and index given by Lemma \ref{lem:prune} for protocol $\pi'$.

  The proof is based on the following observation. Suppose that $x$ is a node in layer $L_1$ and $y$ is a node in layer $L_2$ of the $i^{th}$ BGI-component of the network. According to the definition of a broadcast protocol, the behavior of node $x$ in a given round depends on its own label, which is fixed, the labels of its neighbors and the sequence of messages received in previous rounds. The source is always a neighbor of $x$, but $x$ may or may not be connected to $y$. However, if $x$ is not connected to $y$, then $x$ cannot transmit the message to $y$ and complete broadcast on this BGI-component, so we focus our attention on the situation when $x$ is connected to $y$.

  Now, from lemma \ref{lem:prune}, we get that the source transmits the same message in round $0$ when protocol $\pi'$ is executed on every network of subset $S$, namely the broadcast message and advice string $\upsilon$. At this point, we would like to conclude that, when connected to $y$, the node $x$ has exactly the same behavior in the executions of protocol $\pi'$ on different networks of subset $S$. However, we still have to consider the messages received by $x$ from the node $y$, which can be transmitted in different rounds on the different networks belonging to subset $S$. But then, we recall that the only spontaneous transmission is made by the source in round 0, and so if $y$ makes a transmission the task of broadcast has already been completed on this BGI-component. Thus, we may conclude that, when connected to $y$, node $x$ has always the same behavior when $\pi'$ executes on every network of $S$, up to the round in which broadcast is completed on component $i$. Next, we present the proof of the lower bound formally.

   Consider the family ${\cal F} = \{F_0,F_1,\ldots,F_{r-1}\}$ of subsets of nodes from the BGI-component $i$, defined as follows. A node $x$ belongs to subset $F_j$ if and only if there exists a network $N \in S$ such that

  \begin{enumerate}
  \item
     $x$ is connected to $y$ in the network $N$
  \item
     $x$ transmits in round $3*j + 1$ when $\pi'$ is executed on network $N$
  \item
     $y$ does not receive any message up till round $3*j + 1$ when $\pi'$ is executed on $N$.
  \end{enumerate}

  Now, we claim that ${\cal F}$ corresponds to a $(\sqrt{n},\sqrt{n})$-selective family. To see this, note that for an arbitrary network $N$, the subset of nodes in layer $L_1$ of the $i^{th}$ BGI-component of network $N$ that are connected to the node $y$, correspond to an arbitrary subset of $[\sqrt{n}]$. Let this subset be called $Z$. Since $\pi'$ completes broadcasting on network $N$ in at most $3*r$ rounds, there must exist some $j < r$ such that in round $3j + 1$, for the first time, exactly one node from subset $Z$ transmits. But this is equivalent to saying that there exists a subset $F_j$ in ${\cal F}$ such that $|F_j \cap Z| = 1$. Finally, to appy the theorem \ref{thm:sel-family}, we restrict ourselves to $(n,k)$-selective families, where $k \leq \frac{n}{64}$. However, it is easy to see that if ${\cal F}$ is $(\sqrt{n},\sqrt{n})$-selective, then it is also $(\sqrt{n},\sqrt{n}/64)$-selective. Now, $r = |{\cal F}| \geq \frac{\sqrt{n}}{64 \cdot 24} * \log{n} \geq \frac{\sqrt{n}}{1536}$. But this contradicts our initial assumption that there exists a deterministic broadcast protocol $\pi$ that completes broadcast in less than $\sqrt{n}/1536$ rounds on every network of $C_2$ family.
\end{proof}

  We note here that better constants could be derived by using the hitting game argument from \cite{BGI}. However, we believe that the current lower bound is not asymptotically close to the optimal and hence have chosen not to optimize the constants. Below, we note a corollary to the above theorem for networks of arbitrary diameter $D$.

\begin{corollary}
  Every deterministic broadcast protocol must take $\Omega(\sqrt{n*D})$ rounds to complete broadcast on radio networks of diameter $D$.
\end{corollary}

\section{Reductions (Proof of the Lemma \ref{lem:reduction})}
\label{sec:reductions}

  In this section we present a series of simulation based reductions to prove Lemma \ref{lem:reduction}.

\subsection{Reduction 1}
\label{subsec:reduction1}

  The first reduction shows that we may consider protocols with a simplified communication structure, at the cost of a constant factor (= 3) increase in the number of rounds. In these protocols, the nodes coordinate their transmissions so that collisions involving nodes from different layers do not occur.

\begin{lemma}
\label{lem:reduction1}
   Assume that there exists a protocol $\pi_0$ that completes broadcast in at most $r$ rounds on every $C_2$ network. Then, there exists a protocol $\pi_1$ that completes broadcast in at most $3*r$ rounds on every network of $C_2$ family, such that the nodes in layer $L_i$ transmit in round $t$ only if $t \equiv i$ (mod 3), for $i = 0, 1, 2$. 
\end{lemma}

\begin{proof}
  Protocol $\pi_1$ simulates each round $t$ of protocol $\pi_0$ in the sequence of rounds $3t$, $3t+1$ and $3t+2$. The idea is that, in round $3t+i$, each node in layer $L_i$ takes the action that it would take in round $t$ under protocol $\pi_0$.

  Assuming that $\pi_1$ includes the description of protocol $\pi_0$, then it is sufficient to show that each node $w$ in the network can compute the list of messages received by itself up to round $t$ under protocol $\pi_0$, during the execution of $\pi_1$.

   The claim is certainly true for $t=0$. Now, suppose that, at the beginning of round $3t$,
  every node in the network has the correct list of messages received up to round $t-1$ under
  protocol $\pi_0$. Then, each node can take the appropriate action during rounds $3t$, 
  $3t+1$ and $3t+2$, and update their list of received messages as follows. If $w$ is a
  node in layer $L_i$, then

  \begin{enumerate}[a)]
  \item
     if $w$ transmits a message in round $3t+i$, it appends $\phi$ to its list of received
    messages, since it cannot receive a message in a round in which it acts as a transmitter.
  \item
     otherwise, if exactly one neighbor of $w$ transmits during rounds $3t$, $3t+1$ and
    $3t+2$, it appends this message to the list; else it appends $\phi$.
  \end{enumerate}

   Condition (a) is obvious. The only issue on condition (b) is the ability of the node $w$
  to detect that a single neighbor transmitted in rounds $3t$, $3t+1$ and $3t+2$. If $w$ 
  belongs to layer $L_0$ or $L_2$, then all its neighbors are in layer $L_1$ and, according
  to the definition of a broadcast protocol, it receives a message in round $3t+1$ if and
  only if exactly one of its neighbors transmit. If $w$ belongs to layer $L_1$, then it has
  the source and possibly a node from layer $L_2$ as neighbors. But since in this case no
  collision is possible, $w$ can detect exactly which of its neighbors transmitted.

   Hence, the conditions above guarantee that, at the beginning of round $3(t+1)$, each
  node in the network has the correct list of messages received under protocol $\pi_0$
  up to round $t$.
\end{proof}

   We denote by $\Pi_1$ the class of protocols satisfying the conditions of lemma 
  \ref{lem:reduction1}.

\subsection{Reduction 2}
\label{subsec:reduction2}
   The second reduction shows that, instead of transmitting arbitrary messages, the source can 
  just retransmit whatever message it received in the previous $L_1$ transmission round. 
  The idea is that the nodes 
  of layer $L_1$ can simulate the behavior of the source, once they have the sequence of messages 
  it has received so far, and in this way to recover the message that the source should have 
  transmitted during the execution of some other protocol.

\begin{lemma}
\label{lem:reduction2}
   Assume that there exists a protocol $\pi_1$ of the class $\Pi_1$, that completes broadcast in at most $3r$ rounds on every $C_2$ network. Then, there exists a protocol $\pi_2$, also of the class $\Pi_1$, that completes broadcast in at most $3r$ rounds on every $C_2$ network such that, in round $3t$, the source just retransmits whatever message it received in round $3t-2$.
\end{lemma}

\begin{proof}
  In protocol $\pi_2$ the nodes of the network behave as follows.

  The source node behaves according to the statement of the lemma. The nodes in layer $L_2$ just simulate protocol $\pi_1$. The nodes in layer $L_1$ decide what action to perform in round $3t+1$ in two steps. 

  First, they use the list of messages received from the source up to round $3t$ to simulate the behavior of the source under protocol $\pi_1$, and obtain the sequence of messages that the source should have transmitted up to round $3t$ in the execution of protocol $\pi_1$. 

  This second list, together with the list of messages received from a possible neighbor in layer $L_2$, allows them to simulate their own behavior under protocol $\pi_1$ and compute the appropriate action to perform in round $3t+1$.
\end{proof}

  We denote by $\Pi_2$ the class of protocols satisfying the conditions of lemma \ref{lem:reduction2}.

\subsection{Reduction 3}
\label{subsec:reduction3}

   This reduction shows that, instead of transmitting arbitrary messages, the source can
  just send descriptions of BGI-components of the network on which the protocol is being
  executed. The idea is that, with the description of a BGI-component and the list of
  messages transmitted by the source,  any node of the network can recover the
  messages transmitted by the nodes of this component to the source, and this is all
  that is required to simulate a protocol of the class $\Pi_2$. 

\begin{lemma}
\label{lem:reduction3}
  Assume that there exists a protocol $\pi_2$ of the class $\Pi_2$ that completes broadcast 
  in at most $3r$ rounds on every $C_2$ network. Then, there exists a broadcast protocol 
  $\pi_3$, of the class $\Pi_1$, that completes broadcast in at most $3r$ rounds on every 
  $C_2$ network such that

  \begin{itemize}
  \item
     the source is provided, as input, with a complete description of the network on which 
    the protocol is being executed.

  \item
     in round $3t$ the source transmits:
    \begin{enumerate}[a)]
    \item
       the broadcast message $\mu$, if $t=0$.
    \item
       $\phi$, if no message is received in round $3t - 2$.
    \item
       $<i,\tau>$, if a message is received in round $3t - 2$; in this case, the tuplet $<i,\tau>$ is a description of the BGI-component of the transmitting node in the network on which the protocol is being executed.
    \end{enumerate}
  \end{itemize}
\end{lemma}

\begin{proof}
   In protocol $\pi_3$ the nodes of the network behave as follows.

   The source node behaves according to the statement of the lemma. The nodes in layer $L_2$ just simulate protocol $\pi_2$. The nodes in layer $L_1$ need, for the simulation of protocol $\pi_2$, to compute, in each round $3t$, the message received by the source in round $3t-2$. We consider two cases:

  \begin{enumerate}
  \item
     If the source transmits $\phi$ in round $3t$, this means that it has received 
    no message in round $3t-2$, either because no neighbor transmitted
    or because a collision occurred.
  \item
     If the source transmits $<i,\tau>$ in round $3t$, this means that a single
    neighbor, from the BGI-component $i$, transmitted in round $3t-2$. To recover
    this message, each node in $L_1$ can simulate the behavior of the nodes in
    component $i$ up to round $3t-2$, using the description $\tau$ and the list
    of messages transmitted by the source up to round $3(t-1)$ under protocol $\pi_2$.
  \end{enumerate}
\end{proof}

   We denote by $\Pi_3$ the class of protocols that satisfy the conditions of lemma \ref{lem:reduction3}

\subsection{Reduction 4}
\label{subsec:reduction4}

  This reduction shows that, instead of providing the source with a complete description of the topology of the network, it is sufficient to provide it with just a partial topological information in the form of an advice string. Moreover, the behavior of the source is further restricted so that now it just transmits the advice string together with the broadcast message in round 0, and remains silent thereafter. The idea is that the advice string gives the sequence of messages transmitted by the source under protocol $\pi_3$, and with this information the nodes in layer $L_1$ can carry out a complete simulation of protocol $\pi_3$ on the network.

\begin{lemma}
\label{lem:reduction4}
   Assume that there exists a protocol $\pi_3$ from class $\Pi_3$, that completes broadcast in at most $3r$ rounds on every network of family $C_2$. Then, there exists a broadcast protocol $\pi_4$, of the class $\Pi_1$, such that 

  \begin{itemize}
  \item
     for each network $N$ in $C_2$ there is an advice string $\upsilon(\pi_3,N,3r)$ such that
    protocol $\pi_4$ completes broadcast on network $N$ in at most $3r$ rounds when the source
    is provided with this advice string in the set up phase.
  \item
     the source transmits the broadcast message and $\upsilon(\pi_3, N, 3r)$ in round $0$.
  \item
     the source remains silent in every round $i > 0$.
  \end{itemize}
\end{lemma}

\begin{proof}
  The content of the advice string $\upsilon(\pi_3,N,3r)$ is just the sequence of messages transmitted by the source node in the first $3r$ rounds of the execution of protocol $\pi_3$ on network $N$.

  In protocol $\pi_4$ the nodes of the network behave as follows. The source node behaves according to the statement of the lemma. The nodes in layer $L_2$ just simulate protocol $\pi_3$. In each round $3i + 1$, the nodes in layer $L_1$ read the $i^{th}$ entry of the advice string $\upsilon(\pi_3, N, 3r)$, to obtain the message that the source would transmit in round $3i$ under protocol $\pi_3$, and then simulate their own behavior under protocol $\pi_3$ to compute the appropriate action to perform in this round.
\end{proof}

   The protocol $\pi_4$ given by Lemma \ref{lem:reduction4} is the protocol $\pi'$ mentioned in the Lemma \ref{lem:reduction} of Section \ref{sec:lower-bound}.

\section{Procedure Prune (Proof of the Lemma \ref{lem:prune})}
\label{sec:prune}

   In this section we describe a procedure to obtain an advice string $\upsilon$ and a corresponding subset of networks $S \subset C_2$ for a protocol $\pi'$ given by Lemma \ref{lem:reduction}, such that the following two conditions hold true:\\

\begin{enumerate}
  \item
     Protocol $\pi'$ completes broadcast in at most $3*r$ rounds on every network of subset $S$, when the source node is provided with the advice string $\upsilon$ in the set up phase.
  \item
     There exists an index $i$ such that, for each possible topology of a BGI-component, there is a network $N \in S$ whose $i^{th}$ component has exactly this topology.
\end{enumerate}

  Recall that protocol $\pi'$ is obtained through a series of reductions, and  let $\pi_3$ be the protocol from which $\pi'$ is obtained in the last reduction. Under protocol $\pi_3$, the source transmits only the broadcast message $\mu$, the message $\phi$, and topological information of BGI-components $<i,\tau>$. The sequence of messages transmitted by the source node under protocol $\pi_3$ (disregarding message $\mu$) gives the advice string $\upsilon(\pi_3,N,3r)$ provided to the source in the execution of protocol $\pi'$.

  The subset of networks $S$ is constructed by the execution of procedure Prune(), and considers the execution of protocol $\pi_3$ on every network in $C_2$.  In the description of Prune() and in the arguments that follow, we use the notation $\varepsilon(\pi_3,M,3t)$ to indicate which event occurs in round $3t$ when $\pi_3$ is executed on network $N$. There are three possibilities:

  \begin{itemize}
  \item
     $\varepsilon(\pi_3,M,3t) = \phi$ indicates that the source transmits $\phi$ in round $3t$ because none of its neighbors transmitted in round $3t-2$.
  \item
     $\varepsilon(\pi_3,M,3t) = \rho$ indicates that the source transmits $\phi$ in round $3t$ because more than one of its neighbors transmitted in round $3t-2$.
  \item
     $\varepsilon(\pi_3,M,3t) = <i,\tau>$ gives the message transmitted by the source in round $3t$ when only one node from layer $L_1$ transmits in round $3t-2$.
  \end{itemize}

{\tt
\noindent \\
Procedure Prune($\pi_3$,r) \{ 
\\ \indent \indent
S := set with all networks of $C_2$ 
\\ \\ \indent \indent
If $r = 1$ Then Return $S$
\\ \\ \indent \indent
For $t := 0$ To $r-1$ Do \{  
\\ \\ \indent \indent \indent
If there exists $N \in S$ with $\varepsilon(\pi_3,N,t) = \rho$ 
\\ \indent \indent \indent \indent
Delete from $S$ every network $M$ with $\varepsilon(\pi_3,M,t) \neq \rho$
\\ \\ \indent \indent \indent
Else, if there exists $N \in S$ with $\varepsilon(\pi_3,N,t) = <i,\tau>$, for some $i$
\\ \indent \indent \indent \indent
Fix an arbitrary such network $N$, and
\\ \indent \indent \indent \indent
Delete from $S$ every network $M$ with $\varepsilon(\pi_3,M,t) \neq \varepsilon(\pi_3,N,t)$
\\ \indent \indent
\}
\\ \indent \indent
Return($S$)
\\
\}
}

  Now, observe that all the networks in the set $S$ returned by Prune() are associated with the same sequence of events $\varepsilon(\pi_3,M,3), \ldots, \varepsilon(\pi_3,M,3(r-1))$. From this, it easily follows that the source transmits the same sequence of messages when protocol $\pi_3$ is executed on every network of the subset $S$. Let $\upsilon$ denote the advice string corresponding to this sequence of messages, then the above condition (1) holds.

  To check that condition (2) also holds, we fix an arbitrary network $N$ from $S$, execute protocol $\pi_3$, on $N$ and mark its components as follows:

  \begin{enumerate}
  \item
     If $\varepsilon(\pi_3,M,3t) = \phi$, then no component is marked.
  \item
     If $\varepsilon(\pi_3,M,3t) = <i,\tau>$, then mark component $i$.
  \item
     If $\varepsilon(\pi_3,M,3t) = \rho$, then choose two of the nodes that transmit
    in round $3t-2$ arbitrarily, and mark the respective components.
  \end{enumerate}

  Let ${\cal B}$ be the set of components marked in this procedure.

\begin{lemma}
\label{lem:marking}
   Let $M$ be a $C_2$ network whose components listed in ${\cal B}$ have the same topology as in network $N$. Then, $M$ belongs to the set $S$ returned by Prune().
\end{lemma}

\begin{proof}
   We prove the lemma by induction on $r$.

   The case of $r=1$ is trivial. For the general case, we first observe that every network $N'$ which satisfies $\varepsilon(\pi_3,N',t) = \varepsilon(\pi_3,N,t)$, for $t = 0, 1, \ldots, r-1$, belongs to the subset $S$. So, we will prove that the network $M$ satisfies all those equations.

   Suppose that the lemma holds for $r=k$. Now let us consider the case of $r=k+1$. By the inductive hypothesis, we have that $\varepsilon(\pi_3,M,t) = \varepsilon(\pi_3,N,t)$, for $t = 0, 1, \ldots, k-1$, so we only need to show that $\varepsilon(\pi_3,M,3k) = \varepsilon(\pi_3,N,3k)$. There are three cases to consider:\\

  \noindent
  {\bf Case 1:} $\varepsilon(\pi_3,N,3k) = \rho$

   In this case, more than one node from $L_1$ transmit in round $3k-2$ when $\pi_3$ is
  executed on network $N$. Let $x$ and $y$ be the nodes chosen in step (3) of the marking
  procedure above, and suppose that $x$ and $y$ belong to the BGI-components $i$ and $j$,
  respectively.  
 
    We claim that $x$ and $y$ also transmit in round $3k-2$ when $\pi_3$ is executed
  on network $M$, which implies that $\varepsilon(\pi_3,M,3k) = \rho$. To check the claim, 
  observe that the BGI-component $i$ and $j$ have the same topologies in networks
  $N$ and $M$. Since the nodes in these components receive the same sequence of
  messages from the source up to round $3(k-1)$, they have identical behaviors in
  rounds $3k-2$ and $3k-1$ when $\pi_3$ is executed on networks $N$ and $M$. This
  follows from the definition of a deterministic broadcast protocol. In particular, nodes $x$
  and $y$ transmit in round $3k-2$ when $\pi_3$ is executed on network $M$.

  \noindent
  {\bf Case 2:} $\varepsilon(\pi_3,N,3k) = <i,\tau>$

  In this case, a single node from $L_1$ transmits in round $3k-2$ when $\pi_3$ is executed on network $N$. Let $x$ be such node, which, according to the message, belongs to the $i^{th}$ BGI-component.

  Now, we claim that $x$ is also the single node to transmit in round $3k-2$ when $\pi_3$ is executed on network $M$. Before proving the claim, we observe that the BGI-component $i$ is marked in the procedure above, and so it has exactly the same topology on networks $N$ and $M$. Thus, if the claim holds, we have $\varepsilon(\pi_3,M,3k) = <i,\tau>$, as required.

  The fact that node $x$ transmits in round $3*k-2$ when protocol $\pi_3$ is executed on network $M$ follows by the same argument given in case 1. Now, suppose that some other node in $L_1$ also transmits in the same round. Then, a collision would occur and the source would transmit $\rho$ in its next transmission round $3*k$. However, by the inductive hypothesis, the network $M$ survives up to the $(k+1)^{th}$ iteration of Prune(), and so in this iteration the network $N$ would be purged from set $S$, which is a contradiction. Hence, $x$ is the only node to transmit in round $3k-2$ when $\pi_3$ is executed on network $M$.

  \noindent
  {\bf Case 3:} $\varepsilon(\pi_3,N,3k) = \phi$

   In this case, there is no transmission in round $3k-2$ when protocol $\pi_3$ is executed on $N$. A similar argument to the one given in case 2 allows to conclude that there is also no transmission in round $3k-2$ when protocol $\pi_3$ is executed on network $M$. Hence, $\varepsilon(\pi_3,M,3k) = \phi$.
\end{proof}

   Finally, condition 2 easily follows from lemma \ref{lem:marking} by observing that, when $r < \sqrt{n}/2$, the set ${\cal B}$ has less than $\sqrt{n}$ components. Thus, if $i$ is a BGI-component which does not belong to ${\cal B}$, then every network $M$ which has exactly the same topology as that of network $N$ excepting the $i^{th}$ component, belongs to the subset $S$.

\noindent {\bf Acknowledgement(s):}
  We would like to thank Dariusz Kowalski and Andrzej Pelc for useful discussions and encouraging us to work on the lower bound. We thank Eli Gafni for some initial collaboration regarding this work, though he has dignifiedly chosen not to be a co-author. We also thank Gunes Ercal and Chen Avin for useful comments and suggestions about the presentation of an earlier draft of this work.



\begin{thebibliography}{02}
{\scriptsize
\bibitem{ANLP}
N. Alon, A. Bar-Noy, N. Linial and D. Peleg.
{\em A lower bound for radio broadcast},
Journal of Computer and System Sciences 43(1991), 290-298.

\bibitem{A}
B. Awerbuch.
{\em A new distributed depth-first-search algorithm},
Information Processing Letters 20 (1985), 147-150.


\bibitem{BGI}
R. Bar-Yehuda, O.Goldreich and A. Itai.
{\em On the time complexity of broadcast in radio networks: an exponential gap between determinism and randomization},
Journal of Computer and System Sciences 45 (1992), 104-126.


\bibitem{BV03}
C. Brito and S. Vaya.
{\em An information theoretic lower bound for broadcasting in radio networks},
manuscript, submitted to 44th Annual IEEE Foundations of Computer Science, FOCS'03.


\bibitem{BGV04}
C. Brito and E. Gafni and S. Vaya.
{\em An information theoretic lower bound for broadcasting in radio networks},
Symposium in Theoretical Aspects of Computer Science, STACS 2004.


\bibitem{BP}
B. Bruschi and M. Del Pinto.
{\em Lower bounds for the broadcast problem in mobile radio networks},
Distributed Computing 10(1997), 129-135

\bibitem{CF}
I. Chlamtac and A. Farago.
{\em Making transmission schedule immune to topology changes in multi-hop packet radio networks},
IEEE/ACM Trans. on Networking 2(1994), 23-29

\bibitem{CK},
B.S. Chlebus and D. Kowalski.
{\em A better wake-up in radio networks},
In Proceedings of Annual Symposium on Principles of Distributed
Computing, PODC 2004.

\bibitem{CR}
A. Czumaj and W. Rytter.
{\em Broadcasting Algorithms in Radio Networks with Unknown Topology},
To appear in Proceedings of the 44th Annual IEEE Symposium on
Foundations of Computer Science, FOCS 2003, Cambridge, MA.

\bibitem{CRM}
B.S. Chlebus and M. Rokicki.
{\em Asynchronous broadcasting in radio networks},
To appear in International Colloquium in Structural Information and
Communication Complexity, (SIROCCO'2004), LNCS 3104, 57-68.


\bibitem{CW}
I. Chlamtac and O. Weinstein.
{\em The wave expansion approach to broadcasting in multihop radio networks},
IEEE Transactions on Communications 39(1991), 426-433.


\bibitem{CGGPR}
B.S. Chlebus, L. Gasieniec, A. Gibbons, A. Pelc and W. Rytter.
{\em Deterministic broadcasting in unknown radio networks},
In Proceedings of 11th Annual ACM-SIAM Symposium on Discrete Algorithms, SODA'99 861-870.

\bibitem{CGOR}
B.S. Chlebus, L. Gasieniec, A. Gibbons, A. Ostlin and J.M. Robson.
{\em Deterministic radio broadcasting.}
In Proceedings 27th International Colloquium on Automata, Languages and Programming, ICALP'2000,LNCS 1853, 717-728.

\bibitem{CGR}
M. Chrobak, L. Gasieniec and W. Rytter.
{\em Fast broadcasting and gossiping in radio networks},
In Proceedings 41st Annual IEEE Symposium on Foundations of Computer Science, FOCS'2000, 575-581.

\bibitem{CGK}
M. Chrobak, L. Gasieniec and D. Kowalski.
{\em Wake up problem in multihop radio networks},
In Proceedings of 15th Annual ACM-SIAM Symposium on Discrete Algorithms, SODA'2004, 985-993.

\bibitem{CMS}
A.E.F. Clementi, A. Monti and R. Silvestri.
{\em Selective families, superimposed codes, and broadcasting on unknown radio networks},
In Proceedings 12th Annual ACM-SIAM Symposium on Discrete Algorithms (SODA'2001), 709-718


\bibitem{CRUZ}
Cruz, R., and Hajek, B.,
{\em A  new upper bound to the throughput of a multi-access broadcast channel.},
IEEE Transactions Information Theory IT-28, 3 (May 1982), 402-405.

\bibitem{MP}
G. De Marco and A. Pelc.
{\em Faster broadcasting in unknown radio networks},
Information Processing Letter 79(2001),53-56.

\bibitem{EFF}
P. Erd$\ddot{o}$s, P. Frankl, and Z. Furedi.
{\em Families of finite sets in which no set is covered by the union of $r$ others},
Israel Journal of Mathematics 51(1985),79-89.

\bibitem{EK}
M. Elkin and G. Kortsartz.
{\em Improved broadcast schedule for radio networks},
In Sixteenth Annual ACM-SIAM Symposium on Discrete Algorithms, SODA 2005.

\bibitem{GM}
I. Gabour and Y. Mansour.
{\em Broadcast in radio networks},
In Proceedings 6th Annual ACM-SIAM Symposium on Discrete Algorithms, SODA 1996,577-585.

\bibitem{GPP}
L. Gasieniec, A. Pelc and D. Peleg.
{\em The wakeup problem in synchronous broadcast systems},
In SIAM Journal on Discrete Mathematics 14, 2001, 207-222.

\bibitem{H}
F.K.Hwang.
{\em The time complexity of deterministic broadcast in radio networks},
Discrete Applied Mathematics 60(1995),219-222.

\bibitem{I}
P. Indyk,
{\em Explicit constructions of selector and related combinatorial
structures, with applications},
In Proceedings 13th Annual ACM-SIAM Symposium on Discreet Algorithms,
SODA'2002, 697-704.

\bibitem{JS}
T. Jurdzinski and G. Stachowiak,
{\em Probabilistic algorithms for the wake up problem in single-hop
radio networks}
In Proceedings of 13th International Symposium on Algorithms and
Computing, ISAAC 2002, LNCS 2518, 525-549.

\bibitem{KS}
W.H.Kauz and R.R.C. Singleton,
{\em Nonrandom binary superimposed codes},
In IEEE Transactions on Information Theory 10(1964), 363-377.

\bibitem{KM}
E. Kushilevitz and Y. Mansour,
{\em An $\Omega(D\log(N/D))$ lower bound for broadcast in radio networks},
SIAM Journal on Computing 27(1998), 702-712

\bibitem{MOLLE}
M. MOLLE,
{\em Unifications and extensions of the multiple access communications
problem.}
Ph.D. Thesis, University of California, Los Angeles, Los Angeles,
Calif., July 1981.

\bibitem{KP02}
Dariusz R. Kowalski and Andrzej Pelc,
{\em Deterministic Broadcasting Time in Radio Networks of Unknown Topology},
In Proceedings 43rd Symposium on Foundations of Computer Science, FOCS'2002.

\bibitem{KP04}
Dariusz R. Kowalski and Andrzej Pelc,
{\em Time of deterministic broadcasting in radio networks with local knowledge},
In SIAM Journal of Computing 33(4): 870-891, 2004.

\bibitem{KP2}
Dariusz R. Kowalski and A. Pelc,
{\em Faster Deterministic Broadcasting in Ad Hoc Radio Networks},
In Proceedings 43rd Symposium Theoretical Aspects of Computer Science, STACS 2003.

\bibitem{KP3}
D. Kowalski and A. Pelc,
{\em Time complexity of radio broadcasting: adaptiveness vs obliviousness and vs randomization vs determinism},
Theoretical Computer Science 333 (2005),355-371

\bibitem{KP4}
D. Kowalski and A. Pelc,
{\em Centralized deterministic broadcasting in undirected multihop radio networks},
In Proceedings of 7th International Workshop on Approximation algorithms for Combinatorial Optimization Problems, APPROX 2004, LNCS 3122,171-182.

\bibitem{BGIE}
{\em Errata regarding "On the Time-Complexity of Broadcast in Radio
Networks: An Exponential Gap Between Determinism and Randomization"},
Dec. 2002, available from
$http://www.wisdom.weizmann.ac.il/~oded/p_bgi.html$
}
\end{thebibliography}
\end{document}